\documentclass[12pt]{article}

\usepackage{mathrsfs}
\usepackage{amsfonts}
\usepackage{tikz}
\usepackage{amssymb,bm}
\usepackage{amsmath}
\usepackage{amsthm}
\usepackage{bbding}
\usepackage{float}

\usetikzlibrary{matrix,shapes,snakes}

\allowdisplaybreaks
\newtheorem{thm}{Theorem}[section]
\newtheorem{defi}[thm]{Definition}

\newcommand{\qedd}{\hspace*{\fill}$\Box$\medskip}

\def\deg{\hbox{\rm{deg\,}}}

\def\tr{\hbox{\rm{tr}}}

\def\wt{\hbox{\rm{wt}}}
\def\AI{\hbox{\rm{AI}}}

\begin{document}

\title{A remark on algebraic immunity of Boolean functions}

\author{Baofeng Wu\thanks{Key
Laboratory of Mathematics Mechanization, AMSS, Chinese Academy of
Sciences,
 Beijing 100190,  China. Email: wubaofeng@amss.ac.cn}, Jia Zheng\thanks{School of Mathematical Sciences, University of Chinese Academy of Sciences, Beijing 100049, China.
Email: zhengjia11b@mail.ucas.ac.cn}
 }
 \date{}

\maketitle

\begin{abstract}
In this correspondence, an equivalent definition of algebraic
immunity of Boolean functions is posed, which can clear up the
confusion  caused by the proof of optimal algebraic immunity of the
Carlet-Feng function and some other functions constructed by virtue
of Carlet and Feng's idea. \vskip .5em


\end{abstract}


\section{Introduction}
\label{intro}

Due to the great success of algebraic attacks improved by Courtois
and Meier to such well-known stream ciphers as Toyocrypt and
LILI-128 \cite{AA03}, the notation of algebraic immunity of Boolean
functions was introduced in \cite{AI04} to measure the ability of
functions used as building blocks of key stream generators resisting
this new kind of attacks. In fact, the algebraic immunity of a
Boolean function is the smallest possible degree of such non-zero
Boolean functions that can annihilate it or its complement. For an
$n$-variable Boolean function, the algebraic immunity of it is upper
bounded by $\lceil\frac{n}{2}\rceil$ \cite{AA03}, and when this
upper bound is attained, it is often known as an algebraic immunity
optimal function, or an OAI function for short.

Constructing OAI functions, especially OAI functions together with
other good cryptographic properties such as balancedness, high
nonlinearity, high algebraic degree, etc., is an important problem
in cryptography thanks to the wide usage of Boolean functions in
stream cipher designs. In recent years, a lot of progress has been
made in this problem, and many OAI Boolean functions satisfying all
other main criteria have been constructed. Among all these
constructions, the one belonging to Carlet and Feng \cite{CF08}
seems to be most important since the clever idea proposed by them of
using univariate representations of Boolean functions and BCH bound
from coding theory in the proof of optimal algebraic immunity of the
constructed functions greatly influenced this field. In fact,
subsequent constructions given by Tu and Deng \cite{TD11}, Tang et
al. \cite{TDT13}, Jin et al. \cite{Jin11} and Zheng et al. \cite{ZW}
all adopt Carlet and Feng's idea.

However, we notice that the Carlet-Feng function, as well as all
subsequent functions, seem to suffer from a common problem in the
proof of optimal algebraic immunity of them. In fact, under
univariate representation, to prove an $n$-variable Boolean function
$f(x)$ have optimal algebraic immunity, the standard technique is to
assume the existence of an annihilator $g(x)$ of algebraic degree
less than $\lceil\frac{n}{2}\rceil$ and deduce all coefficients of
$g(x)$ are zero. But following this technique to prove optimal
algebraic immunity of $f(x)$ when it is the function they
constructed, Carlet and Feng neglected that the assumed $g(x)$ that
could annihilate $f(x)$ or $f(x)+1$ should be a Boolean function. In
fact, they proved that for any polynomial
$g(x)\in\mathbb{F}_{2^n}[x]/\langle x^{2^n}+x\rangle$ of algebraic
degree less that $\lceil\frac{n}{2}\rceil$, it must be null if it
could annihilate $f(x)$ or $f(x)+1$. Of course, this leads to the
optimal algebraic immunity of $f(x)$, but it seems that the
properties of $f(x)$ are stronger. In other words, it seems that
Carlet and Feng's construction can be generalized such that no
non-zero Boolean annihilator but possibly  some polynomial
annihilators with degree less that $\lceil\frac{n}{2}\rceil$ of the
constructed functions exist, which can also promise optimal
algebraic immunity of them.

Regretfully,  the above idea cannot be realized. In this
correspondence, we prove that a Boolean function has no non-zero
Boolean annihilator of degree less that $d$  if and only if it has
no polynomial annihilator of degree less that $d$. As a result, the
standard definition of algebraic immunity of Boolean functions can
be modified to an equivalent version, which is given in Section 2.


\section{An equivalent definition of algebraic immunity}
\label{sec:1}

Denote by $\mathbb{B}_n$ the $\mathbb{F}_2$-algebra formed by all
$n$-variable Boolean functions and let $f\in\mathbb{B}_n$. We
firstly recall the standard definition of algebraic immunity of $f$.

\begin{defi}\label{ai}
$$\AI(f)=\min_{0\neq g \in \mathbb{B}_{n}}\{\deg(g)\mid fg=0 ~\text{ or
}~(f+1)g=0\}.$$
\end{defi}

Now we consider the $\mathbb{F}_2$-algebra
$\mathbb{B}_n\otimes\mathbb{F}_{2^n}$. Since
\[\mathbb{B}_n\cong\mathbb{F}_2[x_1,\ldots,x_n]/\langle x_1^2+x_1,\ldots,x_n^2+x_n\rangle,\]
we have
\begin{eqnarray}\label{iso1}
\nonumber\mathbb{B}_n\otimes\mathbb{F}_{2^n}&\cong&\mathbb{F}_2[x_1,\ldots,x_n]/\langle
x_1^2+x_1,\ldots,x_n^2+x_n\rangle \otimes\mathbb{F}_{2^n}\\
&\cong&\mathbb{F}_{2^n}[x_1,\ldots,x_n]/\langle
x_1^2+x_1,\ldots,x_n^2+x_n\rangle.
\end{eqnarray}
Under isomorphism \eqref{iso1}, elements of
$\mathbb{B}_n\otimes\mathbb{F}_{2^n}$ can be viewed as $n$-variable
polynomials over $\mathbb{F}_{2^n}$ reduced modulo $\langle
x_1^2+x_1,\ldots,x_n^2+x_n\rangle$. On the other hand, for any
$F(x_1,\ldots,x_n)\in\mathbb{F}_{2^n}[x_1,\ldots,x_n]/\langle
x_1^2+x_1,\ldots,x_n^2+x_n\rangle$, it can induce a map from
$\mathbb{F}_{2}^n$ to $\mathbb{F}_{2^n}$, which induces a map from
$\mathbb{F}_{2^n}$ to $\mathbb{F}_{2^n}$ due to the existence of the
natural isomorphism $\mathbb{F}_{2}^n\cong\mathbb{F}_{2^n}$. Thus
$F(x_1,\ldots,x_n)$ corresponds to a polynomial in
$\mathbb{F}_{2^n}[x]/\langle x^{2^n}+x\rangle$. It is easy to see
that this correspondence promises an isomorphism between two
$\mathbb{F}_2$-algebras, i.e.
\begin{equation}\label{iso2}
\mathbb{F}_{2^n}[x_1,\ldots,x_n]/\langle
x_1^2+x_1,\ldots,x_n^2+x_n\rangle\cong\mathbb{F}_{2^n}[x]/\langle
x^{2^n}+x\rangle,
\end{equation}
through comparing dimensions of them. Thus we are clear that
\begin{equation}\label{iso3}
\mathbb{B}_n\otimes\mathbb{F}_{2^n}\cong\mathbb{F}_{2^n}[x]/\langle
x^{2^n}+x\rangle.
\end{equation}
That is to say, elements of $\mathbb{B}_n\otimes\mathbb{F}_{2^n}$
can also be distinguished with polynomials over $\mathbb{F}_{2^n}$
reduced modulo $(x^{2^n}+x)$. In this sense, $\mathbb{B}_n$ can be
viewed as an $\mathbb{F}_{2}$-subalgebra of
$\mathbb{B}_n\otimes\mathbb{F}_{2^n}$.

For any $g\in\mathbb{B}_n\otimes\mathbb{F}_{2^n}$, we define its
algebraic degree, denoted $\deg g$, to be the algebraic degree of
the elements in $\mathbb{F}_{2^n}[x_1,\ldots,x_n]/\langle
x_1^2+x_1,\ldots,x_n^2+x_n\rangle$ or $\mathbb{F}_{2^n}[x]/\langle
x^{2^n}+x\rangle$ respectively that are distinguished with $g$ under
the isomorphism \eqref{iso1} and \eqref{iso3} respectively (recall
that for any
$G(x)=\sum_{i=0}^{2^n-1}a_ix^i\in\mathbb{F}_{2^n}[x]/\langle
x^{2^n}+x\rangle$, its algebraic degree is defined as
\[\deg G=\max\{\wt(i)\mid a_i\neq 0\},\]
where ``$\wt(\cdot)$" represents the number of 1's in the binary
expansion of  a non-negative integer \cite{carletBFbook}). Note that
this will not cause confusion since the algebraic degree of any
$H(x_1,\ldots,x_n)\in\mathbb{F}_{2^n}[x_1,\ldots,x_n]/\langle
x_1^2+x_1,\ldots,x_n^2+x_n\rangle$ and the algebraic degree of the
$G(x)\in\mathbb{F}_{2^n}[x]/\langle x^{2^n}+x\rangle$ corresponding
to it under the isomorphism \eqref{iso2} coincide. A short proof of
this fact is like this: under isomorphism \eqref{iso2}, there exists
a basis $\{\beta_1,\ldots,\beta_n\}$ of $\mathbb{F}_{2^n}$ over
$\mathbb{F}_{2}$, such that
\begin{eqnarray*}
 H(x_1,\ldots,x_n)  &=& \sum_{i=0}^{2^n-1}a_i\left(\sum_{i=1}^{n}x_i\beta_i\right)^i \\
   &=&\sum_{i=0}^{2^n-1}a_i\left(\sum_{i=1}^{n}x_i\beta_i\right)^{\sum_{j=0}^{n-1}i_j2^j}  \\
   &=&\sum_{i=0}^{2^n-1}a_i\prod_{j=0}^{n-1}\left(\sum_{i=1}^{n}x_i\beta_i^{2^j}\right)^{i_j}
\end{eqnarray*}
if we assume $G(x)=\sum_{i=0}^{2^n-1}a_ix^i$. Hence
\begin{eqnarray*}
 \deg G  &=&\max\{\wt(i)\mid a_i\neq 0\}  \\
   &=&\max\left\{\left.\sum_{j=0}^{n-1}i_j\,\right|\, a_i\neq 0\right\}  \\
   &\geq&\deg H.
\end{eqnarray*}
Let $\mathcal{A}_d=\{F(x)\in\mathbb{F}_{2^n}[x]/\langle
x^{2^n}+x\rangle\mid \deg F\leq d\}$ and
$\mathcal{B}_d=\{F(x_1,\ldots,x_n)\in\mathbb{F}_{2^n}[x_1,\ldots,x_n]/\langle
x_1^2+x_1,\ldots,x_n^2+x_n\rangle\mid \deg F\leq d\}$ for any $0\leq
d\leq n$. Then it is obvious that $\mathcal{A}_d$ and
$\mathcal{B}_d$ are both vector spaces over $\mathbb{F}_{2^n}$ of
dimension $\sum_{k=0}^{d}{n\choose k}$. It is straightforward that
$\deg G=\deg H$.

Now we give the definition of modified algebraic immunity of the
Boolean function $f$, denoted by $\overline{\AI}(f)$.

\begin{defi}\label{gai}
$$\overline{\AI}(f)=\min_{0\neq g \in \mathbb{B}_{n}\otimes\mathbb{F}_{2^n}}\{\deg(g)\mid fg=0 ~\text{ or
}~(f+1)g=0\}.$$
\end{defi}

It is obvious that $\overline{\AI}(f)\leq\AI(f)$. In the following,
we devote to establish the equivalence between Definition \ref{ai}
and Definition \ref{gai}. For
$H(x_1,\ldots,x_n)\in\mathbb{F}_{2^n}[x_1,\ldots,x_n]/\langle
x_1^2+x_1,\ldots,x_n^2+x_n\rangle$ and a basis
$\{\beta_1,\ldots,\beta_n\}$ of $\mathbb{F}_{2^n}$ over
$\mathbb{F}_{2}$, it is easy to see that there exist
$h_1,~h_2,~\ldots,~h_n\in\mathbb{F}_{2}[x_1,\ldots,x_n]/\langle
x_1^2+x_1,\ldots,x_n^2+x_n\rangle$ such that
$H=\sum_{i=1}^nh_i\beta_i$. It is clear that
\[\deg H=\max\{\deg h_i\mid{1\leq i\leq n}\}.\]
Besides, for the $G(x)\in\mathbb{F}_{2^n}[x]/\langle
x^{2^n}+x\rangle$ corresponding to $H$ under the isomorphism
\eqref{iso2}, there exist
$g_1(x),~g_2(x),~\ldots,~g_n(x)\in\mathbb{F}_{2^n}[x]/\langle
x^{2^n}+x\rangle$ which are all Boolean functions such that
$G=\sum_{i=1}^ng_i\beta_i$ (in fact, $g_i(x)=\tr(\beta_i^*G(x))$ for
$1\leq i\leq n$ where $\{\beta_1^*,\ldots,\beta_n^*\}$ is the dual
basis of $\{\beta_1,\ldots,\beta_n\}$ and ``$\tr$" is the trace map
from $\mathbb{F}_{2^n}$ to $\mathbb{F}_{2}$), and we also have
\[\deg G=\max\{\deg g_i\mid{1\leq i\leq n}\}\]
since $\deg G=\deg H$. With all these preparations, we  can obtain
the following theorem.

\begin{thm}\label{eqi}
For any $f\in\mathbb{B}_n$, $\AI(f)=\overline{\AI}(f)$.
\end{thm}
\begin{proof}
We use the univariate representation of $f$. We need only to prove
$\overline{\AI}(f)\geq\AI(f)$. Assume
$g(x)\in\mathbb{F}_{2^n}[x]/\langle
x^{2^n}+x\rangle\cong\mathbb{B}_{n}\otimes\mathbb{F}_{2^n}$
satisfies $gf=0$. For any basis $\{\beta_1,\ldots,\beta_n\}$ of
$\mathbb{F}_{2^n}$ over $\mathbb{F}_{2}$, there exist
$g_1(x),~g_2(x),~\ldots,~g_n(x)\in\mathbb{B}_{n}$  such that
$g=\sum_{i=1}^ng_i\beta_i$. Hence $\sum_{i=1}^ng_if\beta_i=0$, which
implies $g_if=0$ for all $1\leq i\leq n$. It follows that
\[\AI(f)\leq \max\{\deg g_i\mid 1\leq i\leq n\}=\deg g.\]
When $g$ satisfies $g(f+1)=0$, we can also get $\AI(f)\leq \deg g$.
Hence we have $\AI(f)\leq\overline{\AI}(f)$.\qedd
\end{proof}

From Theorem \ref{eqi}, we are clear that, for any Boolean function
$f\in\mathbb{B}_n$, it has no annihilator in $\mathbb{B}_n$ of
degree less that $d$ if and only if it has no annihilator in
$\mathbb{B}_n\otimes\mathbb{F}_{2^n}$ of degree less that $d$ for
any $1\leq d\leq n$. It can also be concluded that to prove a
Boolean function $f\in\mathbb{B}_n$ have optimal algebraic immunity,
we need  only to prove that there exists no
$g(x)\in\mathbb{F}_{2^n}[x]/\langle x^{2^n}+x\rangle$ with algebraic
degree less that $\lceil\frac{n}{2}\rceil$ such that $gf=0$ or
$g(f+1)=0$ if $f$ is represented by a univariate polynomial over
$\mathbb{F}_{2^n}$. This clears up our confusion with the proofs of
optimal algebraic immunity of the Carlet-Feng function and some
other functions constructed subsequently.

Similarly to the modification of the definition of algebraic
immunity, we can also modify other definitions related to AI of
Boolean functions. For instance, we can give a new and equivalent
version of definition of PAI functions studied in \cite{Mliu12}.

\begin{defi}
Let $f\in\mathbb{B}_n$. $f$ is said to be perfect algebraic immune
if for any positive integers $e<\frac{n}{2}$, the product $gf$ has
degree at least $n-e$ for any non-zero element
$g\in\mathbb{B}_n\otimes\mathbb{F}_{2^n}$ of degree at most $e$.
\end{defi}



\begin{thebibliography}{50}



\bibitem{carletBFbook}
Carlet C.: Boolean functions for cryptography and error correcting
codes. In: Crama Y., Hammer P. (eds.)  Monography Boolean Methods
and Models. Cambridge University Press, London (2010).

\bibitem{CF08}
Carlet C., Feng K.: An infinite class of balanced functions with
optimal algebraic immunity, good immunity to fast algebraic attacks
and good nonlinearity. In: Pieprzyk J. (eds.) ASIACRYPT 2008. LNCS,
vol. 5350, pp. 425--440. Springer, Heidelberg (2008).

\bibitem{AA03}
Courtois N., Meier W.: Algebraic attack on stream ciphers with
linear feedback. In: Biham E. (eds.) EUROCRYPT 2003. LNCS, vol.
2656, pp. 345--359. Springer, Heidelberg (2003).



\bibitem{Jin11}
Jin Q., Liu Z., Wu B., Zhang X.: A general conjecture similar to T-D
conjecture and its application in constructing Boolean functions
with optimal algebraic immunity. Cryptology eprint Archive, Report
2011/515 (2011), http://eprint.iacr.org/



\bibitem{Mliu12}
Liu M., Zhang Y., Lin D.: Perfect algebraic immune functions. In:
Wang X., Sako K. (eds.) ASIACRYPT 2012. LNCS, vol. 7658, pp.
172--189. Springer, Heidelberg (2012).


\bibitem{AI04}
Meier W., Pasalic E., Carlet C.: Algebraic attacks and decomposition
of boolean functions. In: Cachin, C., Camenisch J. (eds.) EUROCRYPT
2004. LNCS, vol. 3027, pp. 474--491. Springer, Heidelberg (2004).


\bibitem{TDT13}
Tang D., Carlet C., Tang X.: Highly nonlinear Boolean functions with
optimum algebraic immunity and good behavior against fast algebraic
attacks. IEEE Trans. Inform. Theory 59, 653--664 (2013).

\bibitem{TD11}
Tu Z., Deng Y.: A conjecture about binary strings and its
applications on constructing Boolean functions with optimal
algebraic immunity. Des. Codes Cryptogr. 60(1), 1--14 (2011).

\bibitem{ZW}
Zheng J., Wu B., Chen Y.,  Liu Z.: Constructing $2m$-variable
Boolean functions with optimal algebraic immunity based on polar
decomposition of $\mathbb{F}_{2^{2m}}$. arXiv: 1304.2946 [cs.CR]

\end{thebibliography}
\end{document}